\documentclass[copyright,creativecommons]{eptcs}
\providecommand{\event}{Gandalf 2011} 
\usepackage{doi}									
\usepackage{breakurl}             
\usepackage[latin1]{inputenc} 
\usepackage{graphicx}  
\usepackage{amsmath,amsfonts,amssymb,amsthm}
\usepackage{paralist}
\usepackage{stmaryrd}
\usepackage{eucal}

\newcounter{blank}

\newtheorem{maintheorem*}[blank]{Main Theorem}

\newtheorem{theorem*}[blank]{Theorem}
\newtheorem{theorem}{Theorem}[section]

\newtheorem{lemma}[theorem]{Lemma}

\newtheorem{definition}[theorem]{Definition}

\newtheorem{remark*}[blank]{Remark}

\newtheorem{note*}[blank]{Note}

\newenvironment{equationarray*}{\begin{equation*}\begin{array}{r@{\ }c@{\ }l}}{\end{array}\end{equation*}}

\renewcommand{\leq}{\leqslant}
\renewcommand{\geq}{\geqslant}

\DeclareMathOperator{\dom}{\mathrm{dom}}

\newcommand{\R}{\mathbb{R}}

\newcommand{\cvE}{\mathcal{E}}

\newcommand{\cvT}{\mathcal{T}}

\newcommand{\cvV}{\mathcal{V}}

\newcommand{\cvO}{\mathcal{O}}
\newcommand{\cvS}{\mathcal{S}}

\newcommand{\tuple}[1]{\langle #1 \rangle}

\newcommand{\autA}{\mathcal{A}}

\DeclareMathOperator{\Type}{{Type}}

\newcommand{\tcat}{\cdot}
\DeclareMathOperator{\Val}{{Val}}
\newcommand{\fval}{\mathit{fval}}
\newcommand{\lval}{\mathit{lval}}
\newcommand{\fstate}{\mathit{fstate}}
\newcommand{\lstate}{\mathit{lstate}}

\newcommand{\ltime}{\mathit{ltime}}

\newcommand{\bx}{\mathbf{x}}
\newcommand{\by}{\mathbf{y}}
\newcommand{\bv}{\mathbf{v}}
\newcommand{\bz}{\mathbf{z}}

\newcommand{\trans}[1]{\xrightarrow{#1}}
\DeclareMathOperator{\trace}{\mathit{trace}}
\DeclareMathOperator{\traces}{\mathit{Traces}}
\DeclareMathOperator{\utrace}{\mathit{utrace}}

\DeclareMathOperator{\trajs}{\mathit{Trajs}}
\DeclareMathOperator{\untime}{\mathit{untime}}
\DeclareMathOperator{\exec}{\mathit{Exec}}

\DeclareMathOperator{\obs}{\mathit{obs}}
\DeclareMathOperator{\Obs}{\mathit{Obs}}

\newcommand{\cut}[1]{ }

\newcommand{\eclass}[1]{\llbracket #1 \rrbracket}

\newcommand{\Th}{\mathfrak{T}}

\newcommand{\clC}{\mathfrak{C}}

\title{A Game-Theoretic approach to Fault Diagnosis of Hybrid Systems\thanks{This research was partly supported by the EU Project FP7-ICT-223844 CON4COORD.}}
\author{Davide Bresolin \quad Marta Capiluppi
\institute{Università di Verona \\
					Dipartimento di Informatica \\
					Verona, Italy}
\email{davide.bresolin@univr.it, marta.capiluppi@univr.it}
}
\def\titlerunning{A Game-Theoretic approach to Fault Diagnosis of Hybrid Systems}
\def\authorrunning{D. Bresolin and M. Capiluppi}
\begin{document}
\maketitle

\begin{abstract}
Physical systems can fail. For this reason the problem of identifying and reacting to faults has received a large attention in the control and computer science communities. In this paper we study the fault diagnosis problem for hybrid systems from a game-theoretical point of view. A hybrid system is a system mixing continuous and discrete behaviours that cannot be faithfully modeled neither by using a formalism with continuous dynamics only nor by a formalism including only discrete dynamics. We use the well known framework of hybrid automata for modeling hybrid systems, and we define a Fault Diagnosis Game on them, using two players: the environment and the diagnoser. The environment controls the evolution of the system and chooses whether and when a fault occurs. The diagnoser observes the external behaviour of the system and announces whether a fault has occurred or not. Existence of a winning strategy for the diagnoser implies that faults can be detected correctly, while computing such a winning strategy corresponds to implement a diagnoser for the system. We will show how to determine the existence of a winning strategy, and how to compute it, for some decidable classes of hybrid automata like o-minimal hybrid automata.
\end{abstract}

\section{Introduction}
In modern complex systems continuous and discrete dynamics interact. This is the case of wide manufacturing plants, agents systems, robotics and physical plants. This kind of systems, called hybrid in their behaviour, need a specific formalism to be analysed. In order to model and specify hybrid systems in a formal way, the notion of \emph{hybrid automata} has been introduced~\cite{Alur95,maler91from}. Intuitively, a hybrid automaton is a ``finite-state automaton'' with continuous variables that evolve according to dynamics characterizing each discrete state. In the last years, a wide spectrum of modeling formalism and algorithmic techniques has been studied in the control and computer science communities to solve the problems of simulation, verification and control synthesis for hybrid systems. Much scarce attention have been posed to the problem of dealing with faults. When a hybrid system fail, the failure propagates throughout the system both in continuous and discrete evolutions. Nevertheless the interaction of continuous and discrete dynamics leads to the need of studying new theories for fault tolerance. 

A \textit{{fault}} is a deviation of the system structure or the system parameters from the
nominal situation \cite{FTC_book}. This implies that after the occurrence of a fault the system
will have a behaviour which is different from the nominal one. Hence \textit{Fault Tolerance} is
the property of reacting to faults. In particular the analysis of fault
tolerance consists in establishing if a given system is still able to
achieve its tasks after the occurrence of a given fault, whereas the synthesis of fault tolerance
resides in providing a given system the tools to react to a given faulty situation. The fault tolerance problem\index{fault!tolerance!problem} can be divided in two tasks: fault
detection\index{fault!detection} and isolation\index{fault!isolation} (FDI\index{FDI}) and control
redesign\index{control redesign}. FDI produces a diagnostic result including detection and location
of the fault, and if possible an estimate of the dimensions of the fault. 
In this paper we concentrate our attention to the problem of fault detection and isolation for hybrid systems.

Fault tolerance and fault tolerant systems have been studied by the control community since the late '70s, as in
\cite{chem_FDI78} where fault detection for chemical processes is introduced, and later in
\cite{FDI81}. One of the first surveys on fault detection is \cite{isermann}, which is dated 1984,
and where some methods based on modelling and estimation are introduced. Much later the interesting
book \cite{patton_book} collects some results on Fault Detection and Isolation (FDI) methods.
For a complete outline of the recent improvements in this field, it
is worth citing \cite{ind_procFD} where a quite
new approach to fault detection in industrial (batch) systems is introduced and \cite{colin_flight},
an overview on fault tolerant techniques for flight control.

In the computer science community fault tolerance is also known as \emph{Fault localization and correction}, and it is usually viewed as the problem of finding and fixing bugs in a software program or in a digital circuit. 
One of the most systematic approaches in this area is Model Based Diagnosis, where an oracle provides an example of correct behavior that is inconsistent with the behavior of the faulty system, and a correct model of the system is usually not necessary~\cite{Console91}. Model based diagnosis can be distinguished between abduction-based and consistency-based diagnosis. Abduction-based diagnosis~\cite{Poole87} assumes that it is known in which ways a component can fail. Using a set of fault models, it tries to find a component and a corresponding fault that explains the observation. Consistency-based diagnosis~\cite{Fey08,Reiter87} considers the faulty behavior as a contradiction between the actual and the nominal behavior of the system. It does not require the possible faults to be known, and it proceeds by dropping the assumptions on the behavior of each component in turn. If this removes the contradiction, the component is considered a candidate for correction.
More recently, applicability of discrete game theory to fault localization and automatic repair of programs have been proposed in~\cite{Jobstmann11}. In this alternative setting, the specification of the correct behaviour is given in Linear Temporal Logic and the correction problem is stated as a game, in which the protagonist selects a faulty component and suggests alternative behaviours.

Not many attempts have been made until now in the field of fault diagnosis
for hybrid systems. This can be due in first instance to the hard task of state estimation in this
kind of systems. Indeed to know if a fault has occurred it has to be detected if the system is
behaving in an unusual way, that is based on the knowledge of the state in which the system is
working. When dealing with hybrid systems a state estimator must provide both the continuous and
the discrete state. The accomplishment of this task is generally difficult because of the coupling
of the two dynamics. 

Among the first methods for fault detection of hybrid systems it is worth citing the ones presented
in \cite{MM_FDI2} and \cite{schroeder}. These two methods are quite different, because they are
based on opposite models of hybrid systems. The first one deals with mixed logical dynamical (MLD)
systems, and mainly with faults on the continuous dynamics, whereas the second one uses quantised
systems, then it deals mainly with the discrete part.
The method introduced in \cite{fourlas01} presents some results based on Hybrid Input/Output Automata 
\cite{Lynch03} and extends the theory of {diagnosability} for discrete events systems to the
hybrid case. As usual in this kind of discrete event approach to hybrid systems, the two dynamics
are kept separated, which means that the diagnoser has to first check if some fault has occurred in
the current (discrete) mode, then to check the continuous dynamics inside the mode, finally a
supervisor will decide which kind of fault has occurred and where. Nevertheless the diagnosability
is tested on the hybrid dynamics, using the notion on hybrid traces.

In this paper we choose to start from the modeling framework of \cite{Lynch03}, where Hybrid Automata assume a distinction between internal and external actions and variables. We add faults to this model, by using a distinguished \emph{fault} action. This is not a restrictive assumption, since every kind of fault can be modeled as an internal action of an automaton, supposing the fault action leads from a nominal state to a faulty one in the system. We assume that  after a fault the system remains in its faulty situation and never recovers. 

We choose to use game theory applied to fault diagnosis of hybrid systems because it allows us not to split the continuous and the discrete behaviours. A hybrid game is a multiplayer structure where the players have both discrete and continuous moves and the game proceeds in a sequence of rounds. In every round each player chooses either a discrete or a continuous move among the available ones~\cite{tomlin2000}. 
Hybrid games has been successfully applied to solve the controller synthesis problem for timed~\cite{Asarin98} and hybrid automata~\cite{Bouyer10,Henzinger99}, and to the fault diagnosis problem for timed automata~\cite{Bouyer05}. In our setting we model the fault diagnosis problem as a game between two players, the environment and the diagnoser. The environment controls the evolution of the system and chooses whether and when a fault occurs. The diagnoser observes the external behaviour of the system and announces whether a fault has occurred or not. Existence of a winning strategy for the diagnoser implies that faults can be identified correctly, while computing such a winning strategy corresponds to implement a diagnoser for the system. 
In contrast with the usual definition of hybrid game, our game is asymmetric, since the environment is more powerful than the diagnoser, and is under partial observability, since the diagnoser is blind to the value of internal variables and to the occurrence of internal events. We define two notions of diagnosability, and we prove that the fault diagnosis problem is solvable for the weakest notion of diagnosability for all classes of hybrid automata that admit a bisimulation with finite quotient that can be effectively computed.

\section{Hybrid Automata with Faults}

Throughout the paper we fix the \emph{time axis} to be the set of non-negative real numbers $\R^+$. An \emph{interval} $I$ is any convex subset of $\R^+$, usually denoted as $[t_1, t_2] = \{t \in \R^+ : t_1 \leq t \leq t_2\}$. For any interval $I$ and $t \in \R^+$, we define $I + t$ as the interval $\{t' + t : t' \in I\}$.

We also fix a countable universal set $\cvV$ of \emph{variables}, where every variable $v \in \cvV$ has a type $\Type(v)$ which defines the domain over which the variable ranges. Elementary types include \emph{booleans}, \emph{integers} and \emph{reals}.  Given a set of variables $V \subseteq \cvV$, a \emph{valuation} over $V$ is a function that associate every variable in $V$ with a value in its type. We often refer to valuation as \emph{states}, and we denote them as $\bx, \by, \bz, \ldots$. The set $\Val(X)$ is the set of all valuations over $X$. Given a valuation $\bx$ and a subset of variables $Y \subseteq X$, we denote the \emph{restriction} of $\bx$ to $Y$ as $\bx|Y$. The  restriction operator is extended to sets of valuations in the usual way. 

A notion that will play an important role in the paper is the one of \emph{trajectory}. A trajectory over a set of variables $X$ is a function $\tau: I \mapsto \Val(X)$, where $I$ is a left-closed interval with left endpoint equal to $0$. With $\dom(\tau)$ we denote the domain of $\tau$, while with $\tau.\ltime$ (the \emph{limit time} of $\tau$) we define the supremum of $\dom(\tau)$. The \emph{first point} of a trajectory is $\tau.\fval = \tau(0)$, while, when $\dom(\tau)$ is right-close, the \emph{last point} of a trajectory is defined as $\tau.\lval = \tau(\tau.\ltime)$. We denote with $\trajs(X)$ the set of all trajectories over $X$. Given a subset $Y \subseteq X$, the \emph{restriction} of $\tau$ to $Y$ is denoted as $\tau|Y$ and it is defined as the trajectory $\tau' : \dom(\tau) \mapsto \Val(Y)$ such that $\tau'(t) = \tau(t)|Y$ for every $t \in \dom(\tau)$. 

A trajectory $\tau'$ is a \emph{prefix} of another trajectory $\tau$ if and only if $\tau'.\ltime \leq \tau.\ltime$ and $\tau'(t) = \tau(t)$ for every $t \in \dom(\tau')$. Conversely, we say that $\tau'$ is a \emph{suffix} of $\tau$ if there exists $t \in \R^+$ such that $\tau'.\ltime = \tau.\ltime - t$ and $\tau'(t') = \tau(t'+t)$ for every $t' \in \dom(\tau')$. Given two trajectories $\tau_1$ and $\tau_2$ such that $\tau_1.\lstate = \tau_2.\fstate$, their concatenation $\tau_1 \cdot \tau_2$ is the trajectory with domain $\dom(\tau_1) \cup (\dom(\tau_2) + \tau_1.\ltime)$ such that $\tau_1 \cdot \tau_2(t) = \tau_1(t)$ if $t \in \dom(\tau_1)$,  $\tau_1 \cdot \tau_2(t) = \tau_2(t-\tau_1.\ltime)$ otherwise. We extend the concatenation operation to countable sequences of trajectories in the usual way.

\medskip

We model hybrid systems with faults by using the formalism of Hybrid Automata (HA) as defined by Lynch, Segala, and Vandraager in~\cite{Lynch03}, enriched with a distinguished \emph{fault} action, and with a partition of the state space into faulty and non-faulty states. We assume a single type of faults for simplicity reasons. However, all the results presented in the paper can be easily generalized to a finite number of faults. 

\begin{definition}\label{def:hioa}
A \emph{Hybrid Automaton with Faults} is a tuple $\autA = \tuple{W, X, Q, Q_f, \Theta, E, H, f, D, \cvT}$, where:

\begin{compactitem}
	\item $W$ and $X$ are two finite sets of \emph{external} and \emph{internal} variables, disjoint from each other.  We define $V = W \cup X$;

	\item $Q \subseteq \Val(X)$ is the set of \emph{states};
	
	\item $Q_f \subset Q$ is the set of \emph{faulty states}. We define $Q_n$  the set of \emph{non-faulty states} such that $Q = Q_n \cup Q_f$ and $Q_n \cap Q_f =\emptyset$.
	
	\item $\Theta \subseteq Q_n$ is a nonempty set of \emph{initial states};
	
	\item $E$ and $H$ are two finite sets of \emph{external} and \emph{internal} actions, disjoint from each other. We define $A = E \cup H$;

	\item $f \in H$ is a distinguished \emph{fault} action;

	\item $D \subseteq Q \times A \times Q$ is the set of \emph{discrete transitions} respecting the following properties: 
		\begin{compactenum}[\bf D1]
			\item for every $\bx \in Q_n$, there exists $\bx'\in Q_f$ such that $(\bx,f,\bx')\in D$;
			\item for every $(\bx, f, \bx') \in D$, $\bx \in Q_n$ and $\bx' \in Q_f$;
			\item for every $(\bx, a, \bx') \in D$ such that $a \neq f$, $\bx \in Q_f$ iff $\bx' \in Q_f$;
		\end{compactenum}

	\item $\cvT$ is a set of trajectories on $V$. Let $\tau.\fstate = \tau.\fval|X$ and  $\tau.\lstate = \tau.\lval|X$, if $\tau$ closed: we require $\cvT$ to respect the following properties:
		\begin{compactenum}[\bf T1]
			\item \emph{faulty state invariance}: for every $\tau$, either $\tau(t)|X \in Q_f$ for every $t\in\dom(\cvT)$, or $\tau(t)|X \in Q_n$ for every $t\in\dom(\cvT)$;
			\item \emph{prefix closure}: for every $\tau'$ prefix of $\tau$, $\tau' \in \cvT$;
			\item \emph{suffix closure}: for every $\tau'$ suffix of $\tau$, $\tau' \in \cvT$;
			\item \emph{concatenation closure}: for every (possibly infinite) sequence of trajectories $\tau_0,\tau_1,\tau_2,\ldots \in \cvT$ such that $\tau_i.\lstate=\tau_{i+1}.\fstate$, the concatenation $\tau_0\tcat \tau_1\tcat \tau_2\tcat \ldots \in \cvT$;
	\end{compactenum}
\end{compactitem}
\end{definition}

\noindent Condition \textbf{D1} implies that a fault can occur at any time of the evolution.
Conditions \textbf{D2} and \textbf{D3} implies that the only discrete action that can switch between non-faulty and faulty states is the fault action $f$, while condition \textbf{T1} implies that trajectories cannot switch between faulty and non-faulty states. Conditions \textbf{T2}, \textbf{T3}, and \textbf{T4} express some natural closure properties on $\cvT$. 

Notice that, following the same approach as Lynch, Segala, and Vandraager, we have defined the state of a Hybrid Automaton with Faults to depend only on the values of the internal variables $X$. 
However, the choice of the set of trajectories $\cvT$ can constrain the admissible values for the external variables in $W$. For this reason, we define the set of \emph{extended states} as $S = \{\bv\in\Val(V) | \exists \tau \in \cvT $ s.t. $ \tau.\fval = \bv\}$. By \textbf{T1} we have that $S|X = Q$, and thus the definition of extended states is sound. The set of \emph{faulty extended states} $S_f$ and the set of \emph{non-faulty extended states} $S_n$ can be defined in a similar way.

\medskip

Given a set of variables $V$ and a set of actions $A$, a \emph{$(V,A)$-sequence} is a possibly infinite sequence $\alpha = \tau_0 a_1 \tau_1 a_2 \tau_2 \ldots$ such that

\begin{compactenum}
	\item $\tau_i$ is a trajectory on $V$, for every $i \geq 0$,
	\item $a_i$ is an action in $A$, for every $i \geq 0$,
	\item if $\alpha$ is finite then it ends with a trajectory, and
	\item if $\tau_i$ is not the last trajectory of $\alpha$, then $\dom(\tau_i)$ is right-closed.
\end{compactenum}

\noindent If $V' \subseteq V$ and $A' \subseteq A$, then the \emph{$(V',A')$-restriction} of $\alpha$ (denoted $\alpha|(V',A')$ is the $(V',A')$-sequence obtained by first projecting all trajectories of $\alpha$ on the variables in $V'$, then removing the actions not in $A'$, and finally concatenating all adjacent trajectories. $(V,A)$-sequences are used to give the semantics of Hybrid Automata in terms of \emph{executions} and \emph{traces}. 

\begin{definition}\label{def:exectraces}
	An \emph{execution} of a Hybrid Automaton $\autA$ from a state $\bx \in Q$ is a $(V,A)$-sequence $\alpha = \tau_0 a_0 \tau_1 a_1 \tau_2 a_2 \ldots$ such that:
	\begin{compactenum}
		\item every $\tau_i$ is a trajectory in $\cvT$;
		\item $\tau_0.\fstate = \bx$;
		\item if $\tau_i$ is not the last trajectory in $\alpha$, then $\tau_i.\lstate \trans{a} \tau_{i+1}.\fstate$, with $a \in A$.
\end{compactenum}
	The corresponding \emph{trace}, denoted $\trace(\alpha)$, is the restriction of $\alpha$ to external variables and external actions. 
\end{definition}

We say that an execution $\alpha = \tau_0 a_0 \tau_1 a_1 \tau_2 a_2 \ldots$ is \emph{faulty} if for some $i \geq 0$, $a_i = f$. An execution $\alpha$ is \emph{maximal} if it starts from a state in $\Theta$ and either it is infinite or its last trajectory $\tau_n$ is such that \emph{(i)} there exists no trajectory $\tau'\in\cvT$ such that $\tau_n$ is a prefix of $\tau'$, and \emph{(ii)} there exists no discrete transition $(\bx,a,\bx')$ with $\bx = \tau_n.\lstate$. Moreover, we say that an execution $\alpha$ is \emph{progressive} if it is infinite and it contains an infinite number of occurrences of external actions. Given a Hybrid Automaton $\autA$, we denote by $\exec(\autA)$ the set of all maximal execution of $\autA$, and by $\traces(\autA)$ the set of all maximal traces of $\autA$, that is, the set $\{\trace(\alpha) : \alpha \in \exec(\autA)\}$. $\autA$ is \emph{progressive} if all executions in $\exec(\autA)$ are progressive.

We say that a hybrid automaton with faults is \emph{diagnosable} if (maximal) faulty executions can be distinguished from non-faulty ones by looking at the corresponding traces.

\begin{definition}[Diagnosability]\label{def:diagnosable}
	We say that a Hybrid Automaton with Faults $\autA = \langle W, X, Q, Q_f, \Theta, E,$ $H, f, D, \cvT \rangle$ is \emph{diagnosable} if for any two maximal executions $\alpha_1, \alpha_2 \in \exec(\autA)$, if $\alpha_1$ is faulty then either $\alpha_2$ is faulty or $\trace(\alpha_1) \neq \trace(\alpha_2)$.
\end{definition}

The above definition of diagnosability is very general, and can be applied to a large class of faults, involving both the continuous and the discrete dynamics of the system. However, solving the fault-diagnosis problem can be very complex, if not impossible at all, under this definition. 

\medskip

In this paper we consider a weaker notion of diagnosability, that we call \emph{time-abstract diagnosability}, for which the fault-diagnosis problem can be solved in a simpler way, leaving the treatment of the stronger diagnosability notion for a subsequent paper.
We assume the system to be progressive, and we define the diagnoser as some kind of finite-state digital device, that monitors the evolution of the system by reacting to external actions and by measuring the values of external variables with a fixed and finite precision. We formally define the latter restriction by introducing the notion of \emph{observation} for the external variables.

\begin{definition}\label{def:observation}
	Given the set of external variables $W$ of a hybrid automaton with faults $\autA$, an \emph{observation} of $W$ is any finite partition $\cvO = \{O_1, \ldots, O_2\}$ of $\Val(W)$. We call the elements $O_i$ of the partition \emph{observables} for $W$.
\end{definition}

In this setting, we say that a progressive system is time-abstract diagnosable if faults can be determined only by looking at the observables and at the occurences of external discrete actions, without considering the delays and the trajectories between them. To formally define such a notion, we first need to define \emph{untimed observation traces} for hybrid automata.

\begin{definition}\label{def:untimed-semantics}
	Given a trace $\beta = \tau_0 a_0 \tau_1 a_1 \tau_2 a_2 \ldots$ of a Hybrid Automaton $\autA$, and an observation $\cvO$ for $W$, we define the corresponding \emph{untimed observation trace} as the sequence $\untime(\beta) = O_0 a_0 O_1 a_1 O_2 a_2 \ldots$ such that $\tau_i.\fval \in O_i$ for each $i \geq 0$. Given an execution $\alpha$ of $\autA$, we define $\utrace(\alpha) = \untime(\trace(\alpha))$.
\end{definition}

\begin{definition}[Time-abstract diagnosability]\label{def:ta-diagnosable}
	We say that a Hybrid Automaton with Faults $\autA = \langle W, X,$ $Q, Q_f, \Theta, E, H, f, D, \cvT \rangle$ is \emph{time-abstract diagnosable} if it is progressive and, for any two maximal executions $\alpha_1, \alpha_2 \in \exec(\autA)$, if $\alpha_1$ is faulty then either $\alpha_2$ is faulty or $\utrace(\alpha_1) \neq \utrace(\alpha_2)$.
\end{definition}

Since $\utrace(\alpha_1) \neq \utrace(\alpha_2)$ implies that $\trace(\alpha_1) \neq \trace(\alpha_2)$, a hybrid automaton that is time-abstract diagnosable is also diagnosable, but the converse does not necessarily hold. Indeed, any fault that do not change the sequence of discrete actions performed by the system, but only the delays or the continuous trajectories between them is not time-abstract diagnosable.

\section{The Fault Detection Game}

In this section we introduce the key notion of \emph{Fault Detection Game} (for time-abstract fault diagnosis), played on a Hybrid Automaton with Faults $\autA$ by two players, the \emph{Environment} and the \emph{Diagnoser}. A \emph{position} in the game is  a pair $(\bv, d) \in \Val(V) \times \{yes, no\}$, such that $\bv$ is an extended state of $\autA$. Given a current position $(\bv, d)$, we distinguish between the following kind of moves:

\begin{enumerate}
	\item \textbf{Diagnoser move}: the Diagnoser chooses an answer $d' \in \{yes, no\}$. The game continues from position $(\bv, d')$ with an Environment move, and we denote this by $(\bv,d) \trans{d'} (\bv,d')$.
	
	\item \textbf{Environment move}: the Environment chooses one of the following possible moves		\begin{compactenum}
			\item two valuations $\bv',\bv'' \in \Val(V)$, a trajectory $\tau \in \cvT$, and an external action $e \in E$ such that $\tau.\fval = \bv$, $\tau.\lval = \bv''$, and $\bv''|X \trans{e} \bv'|X$. The game continues from position $(\bv', d)$ with a Diagnoser move, and we denote this by $(\bv,d) \trans{e} (\bv',d)$; 
			
			\item two valuations $\bv',\bv'' \in \Val(V)$, a trajectory $\tau \in \cvT$, and an internal action $h \in H$ such that $\tau.\fval = \bv$, $\tau.\lval = \bv''$, and $\bv''|X \trans{h} \bv'|X$. The game continues from position $(\bv', d)$ with an Environment move, and we denote this by $(\bv,d) \trans{h} (\bv',d)$.
\end{compactenum}
\end{enumerate}

Notice that the Fault Detection Game is is asymmetric: in our framework the environment is more powerful than the diagnoser, since it can choose the continuous trajectory to follow and prevent the diagnoser to move by choosing an internal action. Moreover, the game is also under partial observability: as formally stated in the following, the diagnoser is blind to the value of internal variables and to the occurrence of internal events.

\begin{definition}[Run of the Fault Detection Game]\label{def:game-run}
A \emph{run} of the game is an infinite sequence $\rho = (\bv_0, d_0) \trans{m_1} (\bv_1, d_1) \trans{m_2} \ldots$ such that:
\begin{compactenum}
	\item $d_0 = no$,
	\item $m_1$ is a diagnoser move,
	\item for every $i \geq 1$, $(\bv_{i-1},d_{i-1}) \trans{m_i} (\bv_{i-1},d_{i-1})$ is a valid move of the game;
	\item for every $i > 1$, $m_i$ is a diagnoser move if and only if $m_{i-1}$ is an environment move with $m_{i-1} \not\in H$.
\end{compactenum}

\noindent A run is \emph{winning} for the diagnoser if one of the two conditions hold:
\begin{compactitem}
	\item either for each $i \geq 1$, $m_i \neq f$ and, for each $j \geq 1$, $d_j = no$, or
	\item there exists $i \geq 1$ such that $m_i = f$ and $j > i$ such that $d_j = yes$.
\end{compactitem}

\end{definition}

Given an observation $\cvO$ for the external variables, the corresponding \emph{observation} of a run $\rho$ is a sequence $\obs(\rho) = (O_0, d_0) \trans{m_1} (O_1, d_1) \trans{m_2} \ldots$ obtained from $\rho$  by replacing every maximal sequence of environment moves $(\bv_j, d_j) \trans{m_{j+1}}	 \ldots \trans{m_{j+k}} (\bv_{j+k}, d_{j+k})$ with $(\bv_j, d_j) \trans{m_{j+k}} (\bv_{j+k}, d_{j+k})$ and by restricting every position $(\bv_j,d_j)$ to $(O_j, d_j)$, where $O_j$ is the unique observable such that $\bv_j|W \in O_j$. We denote by $\Obs_f(\autA)$ the set of finite observations for the Fault Detection Game played on $\autA$. A strategy is a function that tells the Diagnoser which move to choose given a finite observation.

\begin{definition}\label{def:strategy}
A \emph{strategy} is a partial function $\lambda$ from $\Obs_f(\autA)$ to $\{yes, no\}$.
\end{definition}

The strategy tells the diagnoser what  answer to give at the current moment. Let $\rho$ be a run of the game, $\sigma = \obs(\rho)$ and let $\sigma_i = (O_0, d_0) \trans{m_1} \ldots  \trans{m_i} (O_i, d_i) \ldots$ be the prefix of $\sigma$ of length $i$. We say that $\rho$ is consistent with the strategy $\lambda$ when, for all $i$, if $\lambda(\sigma_i)= d$ then either $m_{i+1} = d$ or $m_{i+1}$ is an environment move. A strategy $\lambda$ is \emph{winning from a state $\bx\in Q$} if for all $\bv$ such that $\bv|X = \bx$, all runs starting in $(\bv,no)$ compatible with $\lambda$ are winning. The set of \emph{winning states} is the set of states from which there is a winning strategy.

We can now define the fault diagnosis problems we will study.

\begin{definition}[Time-abstract Diagnosability in a class $\clC$ of automata]\label{def:diagnosability}
	Given a hybrid automaton with faults $\autA \in \clC$, determine whether there exists a winning strategy in the Fault Detection Game played on $\autA$ from the initial states $\Theta$.
\end{definition}

\begin{definition}[Time-abstract Diagnoser synthesis in a class $\clC$ of automata]\label{def:diagnoser-synthesis}
	Given a hybrid automaton with faults $\autA \in \clC$, determine whether there exists a winning strategy in the Fault Detection Game played on $\autA$ from the initial states $\Theta$, and compute such a strategy if possible.
\end{definition}

\section{Computing Strategies}

In this section we will show how to solve the Time-abstract Diagnosability and the Time-abstract Diagnoser synthesis problems for some relevant classes of hybrid automata, exploiting the notion of bisimulation. Such a key notion has been introduced in many fields with different purposes (for instance, van Benthem proposed it 
as an equivalence principle between structures~\cite{vanBenthem78}). In our setting, we use bisimulation as an equivalence principle between states of a hybrid automaton. Roughly speaking, two extended states $\bv$ and $\bv'$ are \emph{bisimilar} if every 
behaviour that starts from $\bv$ can be matched by starting from $\bv'$ and vice versa.

\begin{definition}[Time-abstract bisimulation]\label{def:bisimulation}
	Given a Hybrid Automaton with Faults $\autA = \langle W, X, Q,$ $Q_f, \Theta, E, H, f, D, \cvT\rangle$, a \emph{time-abstract bismulation} is an equivalence relation $\sim \subseteq S \times S$ such that for every $\bv_1, \bv_1', \bv_2 \in S$, the following two conditions are satisfied:
\begin{multline*}
\forall a \in A, \left(\bv_1 \sim \bv_1' \text{ and } \bv_1|X \trans{a} \bv_2|X \right) 
		\Rightarrow 
			\left( \exists \bv_2' \in S \text{ s.t. } \bv_2 \sim \bv_2' \text{ and } \bv_1' \trans{a} \bv_2' \right)
	\text{, and}
\\
\shoveleft{\forall \tau \in \cvT, 
		\left(\bv_1 \sim \bv_1' \text{ and } \bv_1 = \tau.\fval \text{ and } \bv_2 = \tau.\lval \right) 
		\Rightarrow}
\\
\shoveright{\left( \exists \tau'\in \cvT, \bv_2' \in S \text{ s.t. } \bv_2 \sim \bv_2' \text{ and } 
				\tau'.\fval = \bv_1' \text{ and } \bv_2' = \tau.\lval \right).}
\\
\end{multline*}
\end{definition}

Given a hybrid automaton $\autA$ and a time-abstract bisimulation $\sim \subseteq S \times S$, we say that two extended states $\bv,\bv'\in S$ are \emph{bisimilar} if and only if $\bv \sim \bv'$. The \emph{equivalence class of $\bv$}, denoted by $\eclass{\bv}_\sim$ is defined as the set $\eclass{\bv}_\sim = \{\bv'\in S | \bv'\sim\bv\}$ (in the following, we will omit the $_\sim$ subscript when clear from the context). A time-abstract bisimulation naturally induces a partition of $S$ into equivalence classes, called \emph{bisimulation quotient of $\autA$}.

\begin{definition}[Bisimulation quotient]\label{def:bisimulation-quotient}
	Given a Hybrid Automaton with Faults $\autA$ and a time-abstract bisimulation $\sim \subseteq S \times S$, the \emph{bisimulation quotient of $\autA$ under $\sim$} is defined as the set $S/_\sim = \left\{\eclass{\bv}_\sim | \bv \in S\right\}$. 	
\end{definition}

A bisimulation $\sim$ has \emph{finite index} if the number of equivalence classes in $S/_\sim$ is finite, and of \emph{infinite index} otherwise. We say that a class $\clC$ of hybrid automata \emph{admits a bisimulation with finite quotient} if for every $\autA \in \clC$ there exists a time-abstract bisimulation $\sim$ with finite index. We say that such quotient can be \emph{effectively computed} if there exists an algorithm that can compute $\sim$ and $S/_\sim$ for every $\autA \in \clC$. In the following we concentrate our attention on the classes of hybrid automata that admits a bisimulation with finite quotient that can be effectively computed, and
we will prove that the Diagnosability and the Diagnoser synthesis problems are decidable in this case.

In the case of hybrid automata with faults, we have that the equivalence classes of a bisimulation respect the partition between faulty and non-faulty states, as formally proved by the following lemma. In the following, we denote with $S_f/_\sim$ the set of equivalence classes of the faulty extended states of $\autA$, and with $S_n/_\sim$ the set of equivalence classes of the non-faulty extended states of the automaton.

\begin{lemma}
Given a hybrid automaton with faults $\autA$ and a time-abstract bisimulation $\sim \subseteq S \times S$, we have that for every $\bv \in S_n$ and $\bv'\in S_f$, $\bv \not\sim \bv'$.
\end{lemma}

\begin{proof}
Suppose by contradiction that there exists $\bv_1 \in S_n$ and $\bv_1'\in S_f$ such that $\bv_1 \sim \bv_1'$. By \textbf{D1} we have that there must exists $\bv_2 \in S_f$ such that $(\bv_1|X,f,\bv_2|X)\in D$. By the definition of bisimulation, this implies that there exists $\bv_2'\in S$ such that $(\bv_1'|X,f,\bv_2'|X)\in D$, in contradiction with \textbf{D2}, since $\bv_1'|X\in Q_f$.
\end{proof}

Given an observation $\cvO$ of $\Val(W)$, we say that a bisimulation $\sim \subseteq S \times S$ \emph{respects $\cvO$} if for every $\bv, \bv' \in S$, $\bv \sim \bv'$ implies that $\bv|W$ and $\bv'|W$ belong to the same observable of $\cvO$. 
From now on we assume that $\sim$ respects the observation of external variables.

\medskip

We are now ready to define the key notion of \emph{state estimator} of a hybrid automaton with faults. Intuitively, a state estimator is a finite automaton that given an untimed observation trace $\beta$ of $\autA$, provides the set of states that can be reached by $\autA$ under all possible executions compatible with $\beta$.

\begin{definition}[State estimator]\label{def:estimator}
	Given a hybrid automaton with faults $\autA = \langle W, X, Q, Q_f, \Theta, E, H, f, D,$ $\cvT\rangle$, an observation $\cvO$ for the external variables, and a bisimulation with finite index $\sim \subseteq S \times S$ that respects $\cvO$, we define the \emph{state estimator of $\autA$} as the transition system $\cvE = \tuple{2^{Q/_\sim}, \Pi, \Delta}$ such that:
\begin{compactenum}[\bf E1]
	\item $2^{S/_\sim}$ is the powerset of $S/_\sim$;
	\item $\Pi \subseteq 2^{S/_\sim}$ is the set of initial states defined as
	
		\centerline{$\Pi = \{\cvS \in 2^{S/_\sim} | \exists O \in \cvO $ s.t. $ \forall \bv \in S, (\bv|X \in \Theta \wedge \bv|W \in O) \Rightarrow \eclass{\bv} \in \cvS\}$;}
		
	\item $\Delta : 2^{S/_\sim} \times  A \times \cvO \mapsto 2^{S/_\sim}$ is the transition function such that $\Delta(\cvS, a, O) = \cvS'$ iff for all finite executions $\alpha=\tau_0 a_0 \ldots a_n \tau_n$ of $\autA$,
	
	\centerline{$\left( a_n = a \wedge \eclass{\tau_0.\fval} \in \cvS \wedge \utrace(\alpha) = O_0 a O	\right) 	\Rightarrow \eclass{\tau_n.\fval} \in \cvS'$.}
\end{compactenum}
\end{definition}

The state estimator is a deterministic automaton, since the transition function associate a unique successor state to every pair of input symbols $(a,O)$. Hence, with a little abuse of notation, we can define the function $\Delta$ on untimed observation traces as follows. Given an untimed observation trace $\beta = O_0 a_0 O_1 a_1 \ldots$ and a state $\cvS \in 2^{S/_\sim}$, we define $\Delta(\cvS, \beta)$ is the sequence of estimator states $\cvS_0 \cvS_1 \ldots$ such that
\begin{inparaenum}[\it (i)]
	\item $\cvS_0 = \cvS$, and
	\item $\cvS_i = \Delta(S_{i-1}, a_{i-1}, O_i)$ for all $i \geq 1$.
\end{inparaenum}
Moreover, we define $\Delta(\beta)=\Delta(\cvS_0,\beta)$, where $\cvS_0$ is the unique state in $\Pi$ such that $\cvS_0 = \{ \eclass{\bv} \in S/_\sim$ s.t. $\bv|X \in \Theta $ and $\bv|W \in O_0\}$. The following lemma proves that the state estimator is correctly defined, and can be seen ad a consequence of the fact that time-abstract bisimulation preserves traces.

\begin{lemma}\label{lem:estimator}
Given a hybrid automaton with faults $\autA$, and a state estimator $\cvE$ for it, let $\beta$ be a finite untimed observation trace of $\autA$, and $\Delta(\beta) = \cvS_0 \cvS_1 \ldots \cvS_n$. Then, for every $\eclass{\bv} \in S/_\sim$, $\eclass{\bv} \in \cvS_n$ if and only if there exists a finite execution $\alpha = \tau_0 a_0 \tau_1 a_1 \ldots a_{m-1} \tau_{m}$ such that $\utrace(\alpha) = \beta$, $\tau_0.\fstate \in \Theta$, and $\tau_m.\fval \in \eclass{\bv}$.
\end{lemma}

\begin{proof}
Let $\beta = O_0 a_0 O_1 a_1 \ldots a_{n-1} O_n$ be a finite untimed observation trace of $\autA$. We prove the lemma by induction on the length of $\beta$.

If $n = 0$ then $\beta = O_0$ and the claim trivially follows from the definition of $\Delta(\beta)$.

If $n > 0$, let $\beta_{n-1} = O_0 a_0 O_1 a_1 \ldots a_{n-2} O_{n-1}$, and suppose by inductive hypothesis that the claim holds for $\beta_{n-1}$. Now, let $\alpha = \tau_0 a_0 \ldots a_{m-1} \tau_{m}$ an execution of $\autA$ such that $\utrace(\alpha) = \beta$ and $\tau_0.\fstate \in \Theta$. By the definition of untimed execution trace, let $\alpha_{n-1} = \tau_0 a_0 \ldots a_{l-1} \tau_{l}$ be the prefix of $\alpha$ such that $\utrace(\alpha_{n-1}) = \beta_{n-1}$, and let $\Delta(\beta_{n-1}) = \cvS_0 \ldots \cvS_{n-1}$. By inductive hypothesis, we have that $\eclass{\tau_l.\fval} \in \cvS_{n-1}$. Consider now the finite execution $\alpha' = \tau_l a_l \ldots a_{m-1} \tau_m$ such that $\alpha = \alpha_{n-1}\alpha'$. By the definition of untimed observation trace, we have that $\utrace(\alpha') = O_{n-1} a_{n-1} O_n$ and thus, by the definition of $\Delta$, that $\eclass{\tau_m.\fval} \in \Delta(\cvS_{n-1}, a_{n-1}, O_n) = \cvS_n$. To prove the converse implication, let $\eclass{\bv} \in \cvS_n$. By definition of $\Delta$, this implies that there exists a finite execution $\gamma = \tau_0 a_0 \ldots a_{m-1} \tau_m$ such that $\eclass{\tau_0.\fval} \in \cvS_{n-1}$, $\utrace(\gamma) = O_{n-1} a_{n-1} O_n$, and $\tau_m.\fval \in \eclass{\bv}$. By inductive hypothesis, we have that there exists a finite execution $\gamma' = \tau_0' a_0' \ldots a_{l-1}' \tau_{l}'$ such that $\utrace(\gamma') = \beta_{n-1}$, $\tau_0'.\fstate \in \Theta$, and $\tau_l'.\fval \in \cvS_{n-1}$. Hence, the finite execution $\zeta = \tau_0' a_0' \ldots a_{l-1}' \tau_0 a_0 \ldots a_{m-1} \tau_m$ is a valid execution of $\autA$ respecting the desired properties.
\end{proof}

Given the partition of the equivalence classes in $S/_\sim$ between faulty and non-faulty ones, we can distinguish between three different kinds of states $\cvS \in 2^{S/_\sim}$ of the state estimator:

\begin{compactitem}
	\item[\bf \qquad faulty states] such that $\cvS \subseteq S_f/_\sim$,
	\item[\bf \qquad non-faulty states] such that $\cvS \subseteq S_n/_\sim$, and
	\item[\bf \qquad indeterminate states] that contains both faulty and non-faulty equivalence classes.
\end{compactitem}

\noindent It turns out that there exists a winning strategy for the diagnoser on the Fault Detection Game played on $\autA$ if and only if there are no loops of indeterminate states reachable from the initial states of the estimator.

\begin{theorem}\label{teo:diagnosability}
		Given a hybrid automaton with faults $\autA$, an observation $\cvO$ for the external variables, and a bisimulation with finite index $\sim \subseteq S \times S$ that respects $\cvO$, we have that there exists a winning strategy for the diagnoser in the Fault Detection Game played on $\autA$ from the initial sates $\Theta$ if and only if there are no loops of indeterminate states reachable from the initial states $\Pi$ of the state estimator for $\autA$.
\end{theorem}

\begin{proof}
Let $\cvE = \tuple{2^{Q/_\sim}, \Pi, \Delta}$ be the state estimator for $\autA$, and suppose that there are no loops of indeterminate states reachable from the initial states $\Pi$. Then we show how to define a winning strategy for the diagnoser in the Fault Detection Game played on $\autA$ from the initial states $\Theta$.
Given a finite observation for the fault diagnosis game $\sigma = (O_0, d_0) \trans{m_1} (O_1, d_1) \trans{m_2} \ldots \trans{m_n} (O_n,d_n)$, we define the corresponding untimed observation trace $\utrace(\sigma) = O_0 a_0 \ldots a_{l-1} O_l$ by removing all diagnoser moves and ignoring the $d_i$ component of the positions. Let $\Delta(O_0 a_0 \ldots a_{l-1} O_l	) = \cvS_0\ldots\cvS_l$. We define the strategy $\lambda$ on $\sigma$ as follows:
$$
\lambda(\sigma) = \left\{
	\begin{array}{ll}
		yes	& \text{if $\cvS_l$ is a faulty state of $\cvE$} \\
		no		& \text{otherwise}
	\end{array}
\right.
$$
Now, let $\rho = (\bv_0, d_0) \trans{m_1} (\bv_1, d_1) \trans{m_2} \ldots$ be an infinite run of the game compatible with $\lambda$, let $\alpha = \utrace(\obs(\rho)) = O_0 a_1 O_1 a_2 \ldots$ be the corresponding infinite untimed observation trace, and let $\Delta(\alpha) = \cvS_0 \cvS_1 \cvS_2 \ldots$. Two cases may arise:

\begin{itemize}
	\item $\rho$ is faulty. Since there are no loops of indeterminate states in $\cvE$, from Lemma~\ref{lem:estimator} we can conclude that there exists $i \geq 0$ such that for every $j \geq i$ $\cvS_j$ is a faulty state of the estimator. Hence, the strategy $\lambda$ is such that there exists $k$ such that $m_k = yes$, and thus $\rho$ is winning for the diagnoser.

	\item $\rho$ is non-faulty. From Lemma~\ref{lem:estimator} we can conclude that all $\cvS_i$ are either non-faulty or indeterminate.	Hence, the strategy $\lambda$ is such that $m_i = no$ for every diagnoser move, and $\rho$ is winning for the diagnoser.
\end{itemize}

\noindent In both cases the diagnoser wins the game, so we can conclude that $\lambda$ is a winning strategy for the diagnoser in the Fault Detection Game on $\autA$ from the initial states $\Theta$.

\medskip

Conversely, suppose that there exists a loop of indeterminate states reachable from $\Pi$ in $\cvE$. This implies that there exist an indeterminate state $\cvS$ and two time-abstract observation traces $\alpha = O_0 a_0 \ldots a_{n-1} O_n$ and $\beta = O_n a_n \ldots a_{m-1} O_m$ such that:
\begin{compactenum}
	\item $\Delta(\cvS_0, \alpha) = \cvS_0 \ldots \cvS_n$ is such that $\cvS_0 \in \Pi$ and $\cvS_n = \cvS$, and
	\item $\Delta(\cvS_n, \beta) = \cvS_n \ldots \cvS_m$ is such that $\cvS_n = \cvS_m = \cvS$ and $\cvS_i$ is an indeterminate state for each $n \leq i \leq m$.
\end{compactenum}

\noindent Now, suppose by contradiction that there exists a winning strategy $\lambda$ for the diagnoser, and consider the infinite time-abstract observation trace $\gamma = \alpha \beta \beta \beta \ldots$. Two cases may arise:

\begin{itemize}
	\item For every finite prefix $\gamma_i$ of $\gamma$, $\lambda(\gamma_i) = no$. By Lemma~\ref{lem:estimator}, since every state in $\beta$ contains a faulty equivalence class, we have that there exists a faulty execution $\alpha$ of $\autA$ such that $\utrace(\alpha) = \gamma$. This implies that it is possible to build an infinite faulty run of the game that is winning for the environment, against the hypothesis that $\lambda$ is winning for the diagnoser. 
	
	\item There exists a finite prefix $\gamma_i$ of $\gamma$ such that $\lambda(\gamma_i) = yes$. By Lemma~\ref{lem:estimator}, since every state in $\beta$ contains a non-faulty equivalence class, we have that there exists a non-faulty execution $\alpha_i$ of $\autA$ such that $\utrace(\alpha_i) = \gamma_i$. This implies that it is possible to build a run of the game that is winning for the environment, against the hypothesis that $\lambda$ is winning for the diagnoser. 
\end{itemize}

\noindent In both cases a contradiction is found, and the thesis is proved.
\end{proof}

Let $\Th$ be a logical theory. If all the components of a hybrid automaton with faults $\autA$ are definable in $\Th$, we say that $\autA$ is \emph{definable in $\Th$}. Moreover, a class of hybrid automata with faults $\clC$ is \emph{definable in $\Th$} if every $\autA \in \clC$ is definable in $\Th$. The previous theorems shows that the state estimator can be used to define a winning strategy for the diagnoser in the Fault Detection Game. However, it does necessarily implies that we can compute such a strategy, since the theory used to define the automata is not necessarily decidable. Moreover, even when $\Th$ is decidable it is not guaranteed that a bisimulation with finite quotient that can be effectively computed. The following theorem states that if some conditions on the considered theory and on the observation of external variables are respected, then Theorem~\ref{teo:diagnosability} provides an algorithmic solution to the diagnosability and the diagnoser synthesis problem. 

\begin{theorem}\label{teo:decidability}
Let $\Th$ be a decidable theory. Let $\clC$ be a class of Hybrid Automata with Faults that can be defined in $\Th$ and such that for every $\autA$ in $\clC$, there exists a bisimulation with finite quotient $\sim$ that can be effectively computed. Then the time-abstract diagnosability problem in the class $\clC$ is decidable for every observation $\cvO$ definable in $\Th$. Moreover, a winning strategy for the diagnoser can be computed, if possible.
\end{theorem}

\begin{proof}
To prove that that both the time-abstract diagnosability and the time-abstract diagnoser synthesis problems are decidable we have to show how to compute a state estimator $\cvE$ for the automaton $\autA$. 

First of all, let $\cvO$ be a definable observation for the external variables, and let $\sim$ a bisimulation with finite quotient for $\autA$. In general, it is not guaranteed that $\sim$ respects $\cvO$. However, since $\cvO$ is definable in $\Th$, and $\Th$ is decidable, we can always refine $\sim$ to a finer bisimulation $\approx$ respecting $\cvO$ by using the bisimulation algorithm given in~\cite{Brihaye04,Henzinger96}. Since both $\cvO$ and $S/_\approx$ are finite sets, to prove that $\cvE$ can be effectively computed it is sufficient to prove that the the transition relation $\Delta$ is computable. Given a state $\cvS$ of the estimator, an action $a \in E$, and an observable $O \in \cvO$, computing the successor state $\Delta(\cvS, a, O)$ can be reduced to a reachability problem on $\autA$. Since it is known that reachability is decidable for all classes of hybrid automata for which there exists a bisimulation with finite quotient that can be effectively computed, then $\Delta$ is computable and there exists an algorithm that can build the state estimator for $\autA$.

Once that the state estimator $\cvE$ has been built, we can use it for solving both the time-abstract diagnosability and the time-abstract diagnoser synthesis problems as follows.

\begin{itemize}
	\item From Theorem~\ref{teo:diagnosability} we know that there exists a winning strategy for the diagnoser in the Fault Detection Game if and only if there are no loops of indeterminate states in $\cvE$. Since the state estimator is a finite automaton, existence of such loops can be determined by computing a depth-first visit of $\cvE$, and thus the time-abstract diagnosability problem is decidable.
	
	\item The proof of Theorem~\ref{teo:diagnosability} shows how the state estimator can be used to define a winning strategy for the diagnoser in the Fault Detection Game. Since the the state estimator can be effectively computed, we have that such a strategy can be computed. 
		\end{itemize}

\noindent Hence, both problems are decidable under the considered assumptions.
\end{proof}

This decidability results is very general: examples of classes of hybrid automata that respects the conditions of Theorem~\ref{teo:decidability} are Timed Automata~\cite{Alur94}, Simple Multirate Automata~\cite{Alur95}, O-minimal Hybrid Automata~\cite{Brihaye04,Lafferriere00}, and STORMED Hybrid Automata~\cite{Vladimerou08}. Hence, for all such classes of systems, the time-abstract diagnosability problem and the time-abstract diagnoser synthesis problem is decidable. Moreover, the discovery of more classes of hybrid automata respecting the conditions of the theorem immediately leads to new classes of systems for which the two fault-diagnosis problems considered in this paper are decidable. 

The complexity of the two problems depends on the size of the bisimulation quotient $S/_\sim$: if $n$ is the number of equivalence classes, then the size of the state estimator $\cvE$ is exponential in $n$. Since computing a depth-first visit on a finite transition system is in LOGSPACE, we have that the time-abstract diagnosability problem is solvable with polynomial space w.r.t. $n$. Theorem~\ref{teo:decidability} proves that solving the time-abstract diagnoser synthesis problem corresponds to compute the state estimator $\cvE$ for the considered system. Hence, this second problem can be solved using an exponential amount of time w.r.t. $n$.

It is worth to notice that for most classes of hybrid automata, like Timed Automata, Initialized Rectangular Automata, and of o-minimal systems, like Pfaffian Hybrid Automata, the number of equivalence classes in $S/_\sim$ is exponential in the size of the automaton. Hence, for those classes the time-abstract diagnosability problem  is in EXPSPACE and the time-abstract diagnoser synthesis problem is in 2-EXPTIME.

\section{Conclusions}

In this paper we studied the fault-diagnosis problem for hybrid systems from a game-theoretical point of view.  We used the formalism of hybrid automata for modeling hybrid systems with faults and to define the notions of diagnosability and time-abstract diagnosability. We focused our attention on time-abstract diagnosability and we defined a Fault Diagnosis Game on hybrid automata with faults between two players, the environment and the diagnoser. Existence of a winning strategy for the diagnoser implies that faults can be identified correctly, while computing such a winning strategy corresponds to implementing a diagnoser for the system. Finally, we shown how to determine the existence of a winning strategy, and how to compute it, for all classes of hybrid automata definable in a decidable theory $\Th$ and such that a bisimulation with finite quotient can be effectively computed, like timed automata and o-minimal hybrid automata.

The results presented in the paper can be extended in many directions. First of all, by considering the stronger notion of diagnosability instead of time-abstract diagnosability. Then, by extending the results also to undecidable classes of hybrid automata, by exploiting abstraction refinement and approximation techniques.
Finally, in the current framework there is no upper bound on the time that elapses between the occurrence of the fault and the detection by the diagnoser. We envision the extension of our approach to reward and priced hybrid games~\cite{Adler05,Bouyer08} as a possible way to provide minimal-delay strategies for the diagnoser.

\bibliographystyle{eptcs}
\bibliography{unified}

\end{document}